\documentclass[11pt]{article}
\usepackage{amsthm,amsmath,amssymb,amsfonts} 
\usepackage{epsfig} 
\usepackage{latexsym,nicefrac,bbm}
\usepackage{xspace}
\usepackage{color,fancybox,graphicx,subfigure,fullpage}
\usepackage[top=1.25in, bottom=1.25in, left=1.25in, right=1.25in]{geometry}
\usepackage{tabularx}
\usepackage{hyperref} 
\usepackage{pdfsync}
\usepackage[boxruled]{algorithm2e}
\usepackage{multicol}
\usepackage{cite,cleveref}
\usepackage{tikz}

\newcommand{\abs}[1]{\lvert #1 \rvert}

\DeclareMathOperator{\signature}{signature}

\newcommand{\supp}{\mathrm{supp}}
\newcommand{\poly}{\mathrm{poly}}

\renewcommand{\S}{\mathcal{S}}
\newcommand{\T}{\mathcal{T}}
\newcommand{\C}{\mathcal{C}}
\newcommand{\oM}{\overline{M}}

\newcommand{\comment}[1]{}

\newcommand{\R}{\mathbb{R}}

\newcommand{\Z}{\mathbb{Z}}

\newcommand{\norm}[1]{\left\lVert#1\right\rVert}

\newcommand{\code}{C}

\DeclareMathOperator{\nullspace}{nullspace}

\newtheorem{theorem}{Theorem}[section]
\newtheorem{lemma}[theorem]{Lemma}
\newtheorem{definition}[theorem]{Definition}

\newtheorem{corollary}[theorem]{Corollary}

\newtheorem{claim}[theorem]{Claim}
\newtheorem{observation}[theorem]{Observation}

\newtheorem{fact}[theorem]{Fact}

\newtheorem*{theorem*}{Theorem}
\newtheorem*{lemma*}{Lemma}

\crefname{theorem}{Theorem}{Theorems}
\crefname{observation}{Observation}{Observations}
\crefname{proposition}{Proposition}{Propositions}
\crefname{claim}{Claim}{Claims}
\crefname{condition}{Condition}{Conditions}
\crefname{example}{Example}{Examples}
\crefname{fact}{Fact}{Facts}
\crefname{lemma}{Lemma}{Lemmas}
\crefname{corollary}{Corollary}{Corollaries}
\crefname{definition}{Definition}{Definitions}
\crefname{remark}{Remark}{Remarks}

\title{\bf On the Number of Circuits in Regular Matroids \\ (with Connections to Lattices and Codes)}

\usepackage{authblk}

\author[1]{ Rohit Gurjar}
\author[2]{Nisheeth K. Vishnoi}
\affil[1]{Indian Institute of Technology Bombay, India}
\affil[2]{Yale University, USA.}
\date{}

\begin{document}

\maketitle

\begin{abstract}
We show that for any regular matroid on $m$ elements and any $\alpha \geq 1$, the number of  $\alpha$-minimum circuits, or circuits whose size is at most an $\alpha$-multiple of the minimum size of a circuit in the matroid
is  bounded by $m^{O(\alpha^2)}.$
This generalizes a result of Karger for the number of $\alpha$-minimum cuts in a graph.
As a consequence, we obtain similar bounds on the number of $\alpha$-shortest vectors in ``totally unimodular lattices" and on the number of $\alpha$-minimum weight codewords in ``regular codes".
\end{abstract}

\newpage
\tableofcontents

\newpage

\section{Introduction}
We study a general  question about the number of  certain structures, with respect to their sizes,
arising in three different settings:  matroids, codes, and lattices. 
More precisely, we are interested in the growth of  the number of circuits in a matroid, the number of codewords in a code,
 and the number of vectors in an integral lattice with respect to their size, weight, and length, respectively. 
These questions have been extensively studied in various forms in areas such as combinatorial optimization, information and coding theory, and discrete geometry
(e.g., see \cite{Kar96,KalaiLinial95,LM18,ConSlo87}). 

In all of these cases, a trivial (and rough) upper bound on the number of these objects of size $k$ is $m^{\poly(k)}$,
 where $m$ is the underlying ground set size in the case of matroids, length for codes, and dimension in the case of lattices. 
 There are also elementary constructions of matroids/codes/lattices where these upper bounds are tight
 (e.g., graphic matroid of the complete graph, the trivial code with distance $1$, the lattice $\mathbb{Z}^m$). 
However, consider the setting where the shortest size of such an object is $r$. 
 The  question of interest is, when the size $k$ is close to the shortest size $r$, 
 whether the number of objects still grows like $m^{\poly(k)}$.
Since circuits are well studied in matroids, this is a natural question in that context.
In coding theory, the motivation to study this question comes from ``list decodability'' (see \cite{Sudan00}),
and in lattices, this relates to the widely studied question of the ``kissing number'' of a lattice packing (see \cite{ConSlo87}). 
As we subsequently explain, these three questions turn out to be intimately connected in certain cases.

\paragraph{Circuits in matroids}
A circuit of a matroid is a minimal dependent set of its ground elements. 
A motivation for the above question is a seminal result of Karger~\cite{Kar93}
who showed that for a cographic matroid, 
 the number of near-minimum circuits -- circuits whose sizes are at most a constant multiple of the minimum size of a circuit--
is bounded by $\poly(m)$, that is, independent of the minimum size.
The circuits of a cographic matroid are the simple cut-sets of the associated graph, and 
 Karger's result was actually presented in terms of the number of near-minimum cuts in a graph.

An analogous result is also known in the dual setting of graphic matroids. 
The circuits of a graphic matroid are simple cycles in a graph. 
 Subramanian~\cite{Sub95}  (building on \cite{TK92}) showed a $\poly(m)$ bound on the number of near-minimum cycles. 
 Quantitatively, the results of Karger and Subramanian  show that in a graphic/cographic matroid, 
 if the shortest circuit has size $r$, then the number of circuits of size at most $\alpha r$, or $\alpha$-minimum circuits, 
 is bounded by $(2m)^{2 \alpha}$.
Subramanian raised the question of identifying other matroids that have only polynomially many near-minimum circuits.
 
Do all matroids have such a property?
The answer is no:  the uniform matroid can have exponentially many shortest circuits.
Consider the uniform matroid $U_{r,m}$, a matroid of ground set size $m$ where every subset of size at most $r$ is independent.
A circuit of $U_{r,m}$ is any subset of size $r+1$. Thus, the number of shortest circuits is $\binom{m}{r+1}$, i.e., exponential in $r$.
Since a uniform matroid is representable (by a family of vectors over some field),
one can also rule out an affirmative answer for all representable matroids.

The next natural candidate to consider would be the class of binary matroids -- matroids representable over $GF(2)$ --
that also contains graphic and cographic matroids.
Circuits of  binary matroids are closely connected to codewords in binary  linear codes
 and have received considerable attention from this perspective.

\paragraph{Binary linear codes}
Let $A$ be a matrix over $GF(2)$ representing a binary matroid $M$ on the ground set $[m]$,
i.e., $A$ has $m$ columns, and a set is independent in $M$ if and only if the corresponding set of columns in $A$ is linearly independent.
Consider the linear code $C$ whose parity check matrix is $A$. 
The codewords of $C$ are the vectors in $\nullspace(A)$ and thus are precisely the disjoint unions of circuits of $M$.
More importantly, the minimum weight of a codeword in $C$ and the minimum size of a circuit in $M$ are the same. 
And thus, any $\alpha$-minimum weight codeword of $C$ comes from  a union of
$\alpha$-minimum circuits of $M$.
Thus, the question arises: do all binary linear codes have a small number of minimum or near-minimum weight codewords?

This question derives interest from the perspective of list decoding. 
Alon~\cite{Alo97} constructed a binary linear code with $2^{\Omega(\sqrt{m})}$ minimum-weight codewords. 
 Kalai and Linial~\cite{KalaiLinial95} studied distance distributions of codes and conjectured that the above number should be
 $2^{o(m)}$ for all binary linear codes.
However, Ashikhmin, Barg and Vl\u{a}du\c{t}~\cite{ABV01} disproved this conjecture by giving an explicit binary code with
 $2^{\Omega(m)}$ minimum weight codewords.
The question is actually much easier to answer when we consider near-minimum weight codewords. 
 Most binary linear codes have exponentially many near-minimum weight codewords.
For a binary code with a random parity check matrix of dimensions $\lambda m \times m$,
  all its codewords have weights in the range  $[m/3, 2m/3]$ with a high probability, when $\lambda > h(1/3) = - (1/3) \log(1/3) - (2/3) \log(2/3)) \approx 0.918$
  (see~\cite{Bar05}).

 In short, we cannot get the desired polynomial bound for all binary matroids or binary linear codes. 
 Can we identify a subclass of binary matroids where this is true?
Let us first briefly see the history of the analogous question in lattices.
 
 \paragraph{Lattices}
 The question of the number of shortest vectors in a lattice has attracted a lot of attention in mathematics. 
This number is also referred to as the ``kissing number'' of a lattice packing of spheres. 
Consider, for example, the lattice $\mathbb{Z}^m$.
The number of shortest vectors in this lattice is simply $2m$.
Moreover, the number of near-shortest vectors -- whose length is at most a  constant multiple of the shortest length -- 
is bounded by $\poly(m)$.

Such a bound does not hold for general lattices. 
It was long conjectured that there exists a lattice packing with an exponentially large kissing number.
This has been proved recently by Serge Vl\u{a}du\c{t}~\cite{Vla19}.
Lower bounds for the number of near-shortest vectors in lattices were known much earlier. %
Ajtai~\cite{Ajt98} showed that for some constants $\epsilon,\delta >0$ and infinitely many integers $m$,
there exists an $m$-dimensional lattice that has at least $2^{m^\epsilon}$ vectors of length  at most $(1+2^{-m^\delta})$ times
the length of the shortest vector.
A polynomial bound on the number of near-shortest vectors could still hold for some special class of lattices.
Characterizing such lattices is an interesting question. 

\subsection{Our results}
We make progress on the above questions about matroids, lattices, and codes by showing that, for a large class of each,  the number of $\alpha$-minimum circuits, vectors, or codewords grow as $m^{\mathrm{poly}(\alpha)}.$
Our main result concerns matroids; the others follow from connections between matroids and lattices, and between matroids and codes.

\paragraph{Near-minimum circuits in matroids}
We answer the above question about $\alpha$-minimum circuits in the affirmative
for an extensively studied subclass of binary matroids -- called regular matroids. 
These are matroids that can be represented by a family of vectors over every field.

\begin{theorem}[\textbf{Number of near-minimum circuits in a regular matroid}]
\label{thm:num-circuits}
Let $M$ be a regular matroid with ground set size $m$.
Suppose that $M$ has no circuits of size at most $r$. 
Then for any $\alpha \geq 1$, the number of circuits in $M$ of size at most $\alpha r$
is bounded by $m^{O(\alpha^2)}$.
\end{theorem}

\noindent
Since graphic and cographic matroids are the two simplest cases of regular matroids, 
 our result significantly generalizes the results of Karger~\cite{Kar93} and Subramanian~\cite{Sub95}.
Moreover, our result also holds for a class more general than regular matroids, namely max-flow min-cut matroids 
(see Section~\ref{sec:techniques}).
This is the most general class of binary matroids, where a natural generalization of the max-flow min-cut theorem continues to hold. 

A recent work~\cite{GTV18} took a step towards answering this question for regular matroids. 
They showed a polynomial upper bound on
  the number of circuits whose size is less than $3/2$  times the minimum size of a circuit.
  However, a serious shortcoming of their work is that the proof breaks down for $\alpha \geq 3/2$. 

\paragraph{List decodability of codes}
Since binary matroids have close connections with binary linear codes,
our \cref{thm:num-circuits}   implies a list decodability result  
for certain special binary linear codes, called regular codes.
A regular code --  defined in \cite{Kashyap08} --
is a binary linear code such that the columns of its parity check matrix represent a regular matroid. 
We get that for a regular code with distance $d$ and for any constant $\alpha$, the number of codewords with Hamming weight at most $\alpha d$
is polynomially bounded. 
This means that $\code$ is $(\alpha d, \poly(m))$-list decodable, for any constant $\alpha$.
This is in contrast to general binary linear codes, which are not list-decodable beyond the minimum distance.

\begin{corollary}[\bf List decodability of regular codes]
\label{cor:regular-codes}
For a regular code $\code \subseteq GF(2)^m$ with distance $d$, and any $\alpha \geq 1$, the number of codewords with Hamming weight at most $\alpha d$
is bounded by $m^{O(\alpha^3)}$.
\end{corollary}

\noindent
To see this, observe that any codeword of weight at most $\alpha d$  comes from
a combination of at most $\alpha$ circuits of $M$, since each circuit has size at least $d$.
Since \ cref {thm: num-circuits} bounds the number of circuits by $ m^ {O (\ alpha^2)} $, the bound on the number of codewords follows.  

\paragraph{Near-shortest vectors in lattices}
In general, matroids are not related to lattices.
However, since regular matroids are also representable over the real field, 
they happen to be connected to certain lattices.
A  result in matroid theory  (see \cite{Oxl06}) states that 
a matroid is regular if and only if it can be represented by the set of columns of a totally unimodular matrix (TUM).
 A matrix (over reals) is a TUM if each of its minors is either $0$, $1$, or $-1$.
 TUMs are of fundamental importance in discrete optimization, as they are related to the integrality of polyhedra. 

We define a lattice corresponding to a TUM, called a totally unimodular lattice.
  For a TUM $A$, the lattice $L(A)$ of the set of integral vectors in $\nullspace(A)$ is said to be a totally unimodular lattice.
$$L(A) := \{v \in \Z^m \mid Av =0\}.$$ 
It turns out that the near-shortest vectors in a totally unimodular lattice  can be related to 
the near-minimum circuits of the associated regular matroid.
And thus, 
our \cref{thm:num-circuits} implies a polynomial upper bound on the number of near-shortest vectors in 
 totally unimodular lattices.
 
 \begin{theorem}[\textbf{Number of near-shortest vectors in TU lattices}]
 \label{thm:lattice}
 Let $A$ be an $n \times m$ totally unimodular matrix. Suppose any nonzero vector $u \in L(A)$ has a length greater than $\lambda$.
  Then for any $\alpha\ge 1$, the number of vectors in $L(A)$ of length at most $\alpha \lambda$ 
 is   $m^{O(\alpha^6)}$.
 \end{theorem}
 \noindent
Here  by the length of a vector we mean its $\ell_2$-norm. Our proof works for  any $\ell_p$ norm ($p\geq 1$),
with appropriate dependence on $p$.

\subsection{Techniques}

A landmark result in the study of regular matroids  was Seymour's  decomposition theorem~\cite{Sey80},
which also implied the first polynomial-time algorithm for testing total unimodularity of a matrix. 
The theorem states that every regular matroid can be \emph{decomposed} into  simpler matroids, 
each of which is graphic, cographic, or a special $10$-element matroid $R_{10}$.
The theorem has found many uses in discrete optimization and is also used to prove some structural results about regular matroids.
To prove any result about regular matroids, a generic approach can be to first prove the corresponding results for the above simpler matroids
and then ``piece'' them together using Seymour's theorem to ``lift'' the result to regular matroids.
The following results about regular matroids are some examples where this approach has been successful:
finding minimum cycle basis~\cite{GH02} and
deciding first order logic properties~\cite{GKO12}.

 Although Seymour's theorem gives a meta-strategy,
  it does not automatically imply a result for regular matroids, given the result for graphic/cographic/$R_{10}$ matroids.
 Each setting, where  one wants to  apply Seymour's theorem to prove something about regular matroids, requires some new ideas.
In fact, in many settings, Seymour's theorem does not work as it is and a strengthening of its statement is required. 
 Indeed, the work~\cite{GTV18} on near-minimum circuits uses a refined version of Seymour's theorem (given by \cite{Tru98}).
There are a few other results on regular matroids that have used a more refined version (see \cite{Tru98}):
 faster algorithm to test total unimodularity~\cite{Tru90},
upper bounding the cycle cover ratio~\cite{LP02}, 
upper bounding the girth~\cite{Kas09},
and approximating the partition function of the ferromagnetic Ising model~\cite{GJ11}.

Some works, solving discrete optimization problems for regular matroids,
 have used a further stronger version of Seymour's theorem.
The strongest form was presented by Chun and Oxley~\cite{CO11} (also see~\cite{DK13}), which
 gives the flexibility to decompose a regular matroid in many different possible sequences.
They used it to determine the collection of unavoidable parallel minors for 3-connected regular matroids.
Dinitz and Kortsarz~\cite{DK13} used this form to solve the matroid secretary problem for regular matroids.
Later, Fomin, Golovach, Lokshtanov, and Saurabh utilized it in designing 
parameterized algorithms for the space cover problem~\cite{FGLS17}
and the spanning circuits problem~\cite{FGLS17a}.
Aprile and Fiorini used it to give a small extended formulation for the independent set polytope~\cite{AF19}.

Our main result (\cref{thm:num-circuits}) also takes advantage of this strongest version of Seymour's theorem. 
One of our novel ideas is to use different sequences of decompositions for the same matroid to upper bound different classes of circuits.
In contrast,  the result of \cite{GTV18}, which only works for  a multiplicative factor smaller than $3/2$,  uses a fixed decomposition tree. 

\subsection{Future directions}
One natural question motivated by our work is to find out which other matroids have only polynomially many near-minimum circuits. 
Analogous to Seymour's decomposition theorem, Geelen, Gerards, and Whittle~\cite{GGW15} have proposed a structure theorem for any proper minor-closed class of matroids representable over a finite field.
Can we use this structure theorem to upper bound the number of near-minimum circuits 
in these matroids? 

Similarly, for what lattices can we prove a polynomial bound on the number of near-shortest vectors?
An interesting candidate to examine would be the lattice $L(A)$ for any matrix $A$ whose entries are $O(1)$. 

\subsection{Organization of the rest of the paper}
Section~\ref{sec:techniques} gives a brief overview of our proof. 
In Section~\ref{sec:matroids}, we introduce some notation and 
describe a refinement of Seymour's Theorem, 
where we have more structure in the decomposition of a regular matroid. 
In Section~\ref{sec:num-circuits}, we use this structured decomposition to upper bound 
the number of near-minimum circuits in regular matroids. 
In Section~\ref{sec:lattice}, we relate circuits of a regular matroid with vectors of a totally unimodular lattice and prove 
\cref{thm:lattice} using \cref{thm:num-circuits}.

\section{Overview of our proof} 
\label{sec:techniques}

  \cref{thm:num-circuits} is already known in the special cases  of graphic and cographic matroids, and in fact, 
  in a stronger form with weights on the elements. 
 For a weight function $w$ on the elements of a matroid, let us define the weight of any circuit $C$  to be $\sum_{e \in C} w(e)$.
\begin{theorem}[\bf Number of near-minimum circuits in a graphic or cographic matroid~\cite{Kar93,Sub95}]
\label{thm:graphic-cographic}
 Let $M$ be a graphic or cographic matroid with ground set size $m \geq 2$ 
 and 
 a weight function on the ground set. 
If every circuit in $M$ has weight more than $r$, then for any $\alpha \geq 1$, the number of circuits in $M$
of weight at most $\alpha r$ is bounded by $(2m)^{2\alpha}$.
\end{theorem}

\noindent
 In Seymour's decomposition theorem for regular matroids, the building blocks -- graphic, cographic and $R_{10}$ matroids 
 -- are composed together via a binary operation on matroids,  called ``$k$-sum'' for $k=1,2,3$.
 One can visualize this in terms of a ``decomposition tree'' -- a binary tree where each node represents a matroid such that 
 each internal node is a $k$-sum of its two children, each leaf node is a graphic, cographic or the $R_{10}$ matroid, 
 and the root node is the desired regular matroid. 
 It is important to note that this decomposition tree is not unique, and one can perform the decomposition in different ways,
optimizing various parameters, e.g., the tree depth.
 
 The $k$-sum $M = M_1 \triangle M_2$ of two matroids $M_1$ and $M_2$ is a matroid  where $M_1$ and $M_2$ 
 interact through a \emph{common set of elements} $S$ (of size $0$, $1$ or $3$) and each circuit of $M$ is obtained by picking  a circuit
each from $M_1$ and $M_2$, and taking their \emph{sum} via elements in $S$.
Thus, one can hope to upper bound the number of circuits in $M$ using such bounds for $M_1$ and $M_2$.
We know a polynomial upper bound on the number of near-minimum circuits for the matroids that are building blocks of Seymour's decomposition theorem:
\cref{thm:graphic-cographic} shows it for graphic and cographic matroids;
and for the matroid $R_{10}$, the bound holds trivially since it has only constantly many circuits. 
 The challenge is to show that the polynomial bound still holds when we compose these matroids together via a sequence of an arbitrary number of  $k$-sum operations.
 
 A natural attempt to get this bound would be to use an induction based on the decomposition tree of a regular matroid. 
 This was precisely the approach \cite{GTV18} took and 
 showed a polynomial upper bound on the number of circuits in a regular matroid whose size is at most $3/2$ times the shortest circuit size.
 However, their proof technique does not generalize to an arbitrary multiplicative factor $\alpha$.

The inductive approach seems unavoidable
if we work with a given decomposition tree,
since it gives a fixed sequence of $k$-sum operations.
To address this, we need several new ideas. 
As noted earlier, the decomposition tree of a regular matroid is not unique.
Our first main idea is to use {\em different}  decomposition trees to upper bound different kinds of circuits. 
For instance, we can use a decomposition $M = (M_1 \triangle M_2) \triangle M_3$ for one kind of circuits and 
use another decomposition $M = M_1 \triangle (M_2 \triangle M_3)$ for the rest. 

To switch between different decompositions, we need an associativity property
for the $k$-sum operations, 
which Chun and Oxley~\cite{CO11} showed.
Using this associativity, they adapted Seymour's  Theorem to give what we call
 an \emph{unordered decomposition tree (UDT)} of a regular matroid $M$, 
which allows us to construct many different possible decomposition trees for $M$ (also see \cite{DK13, Kap14}).
We believe that using this more structured version of Seymour's theorem is 
necessary to obtain a result corresponding to an arbitrary multiplicative factor $\alpha$.

 For a regular matroid $M$, its UDT is an undirected tree where each node corresponds to a graphic/cographic/$R_{10}$ matroid.
Each subtree $T$ of the UDT corresponds to a unique regular matroid $M_T$.
Moreover, if  $T_1$ and $T_2$ are the two subtrees obtained by deleting an arbitrary edge from a subtree $T$, then 
$$M_T = M_{T_1}\triangle M_{T_2}.$$
Thus, a UDT gives many ways to decompose $M$, depending on the order in which we delete edges from the UDT.
In the proof of \cref{thm:num-circuits}, 
instead of an inductive argument, we use the structure of the unordered decomposition tree directly. 
We give a brief outline of the steps in our proof.
Assume that the shortest size of a circuit in $M$ is more than $r$, 
and we would like to bound the number of circuits of size at most $\alpha r$.

\begin{itemize}
\item \textbf{Circuit decomposition} (Section~\ref{sec:num-circuits}). We first argue that for any partitioning of the UDT into a set of subtrees,
a circuit  $C$ in $M$ can be written as a sum of its \emph{projections} on these subtrees. %
The projections themselves are circuits in the corresponding matroids (\cref{obs:CuCv}).
Importantly, the size of the circuit $C$ equals the sum of the sizes of its projections.
\item \textbf{Balanced partition of UDT} (Section~\ref{sec:partitioning}). For any given circuit $C$ of size $\leq \alpha r$ in $M$, we partition the UDT into a balanced set of subtrees, i.e., we find a set of $O(\alpha)$ \emph{center nodes} in the UDT
such that each subtree obtained by deleting the center nodes has a projection of $C$ of size at most $r/2$ (\cref{cla:centers}).
\item \textbf{Classifying the circuits} (Section~\ref{sec:classifying}). We then classify these circuits based on the partition of the UDT they produce and argue that 
the number of different classes is polynomially bounded (\cref{cla:num-signatures}). 
To bound the number of circuits in a given class, we work with the corresponding UDT partition. 
\item \textbf{Uniqueness of small projections} (Section~\ref{sec:givenSignature}).
The number of circuits in a class is bounded by the \emph{product} of 
the numbers of possible projections on the center nodes and the remaining subtrees.
We first show that if a subtree has a projection of size at most $r/2$,
then the projection is \emph{unique} for all the circuits in the class (\cref{cla:T1}).
\item \textbf{Number of projections on the center nodes} (Section~\ref{sec:givenSignature}).
The main technical step is to bound the number of possible projections on the center nodes. 
 Since there are only $O(\alpha)$ center nodes, it suffices to do it for a given center node.
This task is reduced to bounding the number of weighted circuits in a graphic, cographic, or $R_{10}$ matroid,
 through various technical ideas  (\cref{lem:numberCu}).
\end{itemize}

\noindent
Similar to regular matroids, Seymour also gave a  decomposition theorem~\cite[(7.6)]{Sey80} for max-flow-min-cut (MFMC) matroids,
where there is an additional constant sized matroid called $F_7$ involved in $1$-sums and $2$-sums (also see~\cite{Tru98}).
Our whole argument about bounding the number of near-minimum circuits in a regular matroid
can be applied as it is to MFMC matroids. 
This gives us the same bound on the number of near-minimum circuits in an MFMC matroid as in \cref{thm:num-circuits}.

\section{Matroid preliminaries}
\label{sec:matroids}
For basic definitions and facts in matroid theory, see Oxley's book on Matroid Theory~\cite{Oxl06}.
It is known that a matroid is completely characterized by its set of circuits. 
We use the notation $M=(E,\C)$ to denote a matroid $M$ with ground set $E$ and
set of circuits $\C$.
Recall that the symmetric difference $C_1 \triangle C_2$ of two sets $C_1$ and $C_2$ is given by 
$(C_1 \setminus C_2 ) \cup (C_2 \setminus C_1)$. 
\begin{fact}
\label{fac:symm-diff-circuits}
For two circuits  $C_1$ and $C_2$ of a binary matroid $M$, their symmetric difference $C_1 \triangle C_2$ is a disjoint union of circuits of $M$.
\end{fact}
\noindent
For a matroid $M$ and an element $e$ in it, $M \setminus e$ denotes the matroid obtained after deleting $e$.
 The circuits of $M \setminus e$ are those circuits of $M$ that do not contain $e$. 
\begin{fact}
\label{fac:closed}
The classes of regular matroids, graphic matroids, and cographic matroids are closed under the deletion operation.
\end{fact}

\subsection*{$k$-sums and Seymour's  Theorem}
To prove \cref{thm:num-circuits}, a crucial ingredient is Seymour's remarkable decomposition theorem for regular matroids. 
Seymour~\cite{Sey80} showed that every regular matroid can be constructed by piecing together 
three special kinds of matroids -- graphic matroids, cographic matroids and a certain matroid $R_{10}$.
$R_{10}$ is the matroid represented over GF(2) by a $5 \times 10$ matrix,
with 10 distinct columns, each of which contains exactly three ones.
The operation involved in this composition is called a $k$-sum, for $k=1$, $2$, or $3$,
which is defined for arbitrary binary matroids. 
For two binary matroids $M_1 = (E_1, \C_1)$ and $M_2 = (E_2, \C_2)$ with 
${E_1 \cap E_2} = S$, 
their $k$-sum $M_1 \triangle M_2$ is another binary matroid defined over the ground set $E_1 \triangle E_2$.
The common set $S$, called the gluing set, needs to have different structures depending on whether $k$ is $1$, $2$, or $3$.
$S$ needs to be empty for a $1$-sum, $S$ should have a single element for a $2$-sum,
and $S$ needs to be circuit of size $3$ in both $M_1$ and $M_2$ for a $3$-sum.
For the exact definition of $k$-sum, see \cite{Sey80}.
The following facts \cite{Sey80} about circuits of $M_1 \triangle M_2$ will be sufficient for our purposes.
\begin{fact}
\label{fac:circuit-in-one}
For $i=1,2$, if $C_i$ is a circuit of $M_i$ that does not contain any elements from $S$, then $C_i$ is a circuit of $M_1 \triangle M_2$.
\end{fact}
\begin{fact}
\label{fac:disjointCircuits}
Let $C_i$ be a disjoint union of circuits of $M_i$ for $i=1,2$.
If $C_1 \triangle C_2$ is a subset of $E_1 \triangle E_2$, then it is a disjoint union of circuits in $M_1 \triangle M_2$.
\end{fact}

\begin{fact}[\textbf{Circuits in a $1$-sum}]
\label{fac:1sum-circuits}
If $M$ is a $1$-sum of $M_1$ and $M_2$ then any circuit of $M$ is either a circuit of $M_1$
or a circuit of $M_2$. 
\end{fact}

\begin{fact}[\textbf{Circuits in a $2$-sum or a $3$-sum}]
\label{fac:sum-circuits}
Let $M$ be a $2$-sum or a $3$-sum of $M_1=(E_1,\C_1)$ and $M_2=(E_2,\C_2)$ and 
let $E_1 \cap E_2  = S$ ($\abs{S} = 1$ or $3$).
Then for any circuit $C$ of $M$, exactly one of the following holds:
\begin{enumerate}
\item  $C \in \C_1$ with $S \cap C = \emptyset$ 
\item $C \in \C_2$ with $S \cap C = \emptyset$
\item there exist unique $e \in S$, $C_1 \in \C_1$, and $C_2 \in \C_2$
 such that
$$S \cap C_1 = S \cap C_2 = \{e\}  \mbox{ and } C = C_1 \triangle C_2.$$
\end{enumerate}
\end{fact}
\noindent
The operations of $1$-sum and $2$-sum are associative, however,
the $3$-sum operation is not associative as defined. 
Chun and Oxley~\cite[Lemma 2.3]{CO11} showed that by doing some small
modifications in the summand matroids, $3$-sum can also be made associative. 
This led them to a stronger version of Seymour's Theorem,
which gives a more structured decomposition of a 3-connected regular matroid.
The decomposition can be easily generalized to any regular matroid (also see \cite{DK13}).
It is given by an \emph{unordered decomposition tree (UDT)}, 
 defined as follows.   

Let $T$ be a tree with its vertex set $V(T)$ such that each vertex $v \in V(T)$ 
has a corresponding binary matroid $M_v$ with the ground set $E_v$.
Further, for any two vertices $u$ and $v$ of the tree, we have the following.
\begin{equation}
\label{eq:tree-edge}
\abs{E_u \cap E_v} = \begin{cases}
			0, 1 \text{ or } 3 & \text{ if there is an edge between } u \text{ and } v,\\ 
			 0 & \text{ otherwise}.
\end{cases}
\end{equation}
In particular, this means that any element is a part of at most two ground sets. 
For any subtree $T'$ of this  tree $T$, we define a binary matroid $M_{T'}$ with its ground set being
$E_{T'} = \triangle_{u \in {V(T')}} E_u$. 
It is defined recursively as follows:
\begin{itemize}
\item if $T'$ is a vertex, say $v$, then $M_{T'} = M_v$.
\item Otherwise, let $e = (u,v)$ be an edge in the tree $T'$. Let $T_1$ and $T_2$
be the two subtrees obtained by removing the edge $e$ from $T'$. 
From (\ref{eq:tree-edge}), it follows that $E_{T_1} \cap E_{T_2} = E_u \cap E_v$. 
Define 
$$M_{T'}= M_{T_1} \triangle M_{T_2}.$$
\end{itemize}

\noindent
Here, the sum $M_{T_1} \triangle M_{T_2}$ is a $1$-sum, $2$-sum or a $3$-sum depending on the size of $E_{T_1} \cap E_{T_2} ={E_u \cap E_v}$.
In the above definition, the choice of edge $e$ in the tree $T'$ is arbitrary.
 Chun and Oxley~\cite[Lemma 2.5, Corollary 2.6]{CO11} showed that the matroid $M_{T'}$ obtained is the same for every choice of $e$.
Aprile and Fiorini~\cite[Theorem 6, Proposition 27]{AF22} showed the same with a stronger property about the cographic components of the decomposition tree.

\begin{lemma}[\cite{Sey80,CO11,AF22}]
\label{lem:uDecomposition}
For any regular matroid $M$, there is an unordered decomposition tree $T$ such that $M=M_T$, 
where each vertex of the tree $T$ is a graphic, cographic or the $R_{10}$ matroid.
Moreover, for any cographic matroid, any 3-circuit involved in a 3-sum operation consists of the 3 edges incident to a degree-$3$ vertex of the corresponding graph. 
\end{lemma}
\noindent
The $k$-sum of two matroids deletes the common elements between their ground sets; hence, any element of any matroid is involved in at most one $k$-sum operation.
In other words, various gluing sets in any matroid $M_v$ would be disjoint.

\section{Bounding the number of circuits in a regular matroid}
\label{sec:num-circuits}

This section is devoted to a proof of \cref{thm:num-circuits}.
Using the unordered decomposition tree of a regular matroid from the previous section,
we bound the number of near-minimum circuits in a regular matroid $M$.
Recall that in \cref{thm:num-circuits} we assume there are no circuits in $M$ of size at most $r$.

Let $T$ be the unordered decomposition tree from \cref{lem:uDecomposition} such that $M= M_T$.
Let $M_v$ be the matroid corresponding to a vertex $v$ in $T$ with ground set $E_v$.
Recall that $E = \triangle_{v \in V(T)} E_v$.
Recall from the definition of an unordered decomposition tree 
that if $T_1$ and $T_2$ are subtrees of $T$ obtained by removing an edge
then we can write $M = M_{T_1} \triangle M_{T_2}$.
By \cref{fac:sum-circuits}, any circuit $C$ of $M$ can be written as $C_{T_1} \triangle C_{T_2}$ such that one of the following holds
\begin{itemize}
\item $C_{T_2}$ is empty and $C_{T_1}$ is a circuit of $M_{T_1}$ that avoids the common elements $E_{T_1} \cap E_{T_2}$,
\item $C_{T_1}$ is empty and $C_{T_2}$ is a circuit of $M_{T_2}$ that avoids the common elements $E_{T_1} \cap E_{T_2}$,
\item $C_{T_1}$ and $C_{T_2}$ are circuits of $M_{T_1}$ and $M_{T_2}$, respectively such that
 both $C_{T_1}$ and $C_{T_2}$ contain a single common element coming from $E_{T_1} \cap E_{T_2}$.
\end{itemize}
Recursively, one can define $C_{T'}$ for any subtree $T'$  of $T$, including the case when $T'$ is a vertex.
And we can write $$C = \triangle_{v \in V(T)} C_v.$$
Note that $C_v$ can be empty for some vertices $v$.
The following observation follows.
\begin{observation}[Projections of a circuit on the vertices of UDT]
\label{obs:CuCv}
For any circuit $C$ of $M$,
\begin{itemize}
\item the set of vertices $u$ of the tree $T$ such that $C_u \neq \emptyset$, form a connected subgraph of $T$,
\item for any two adjacent vertices $u$ and $v$ in the tree $T$, if $C_u$ and $C_v$ are non-empty then
  $C_u$ and $C_v$ both contain an element from $E_u \cap E_v$, which is the same for them,
 \item for any two non-adjacent vertices $u$ and $v$ in the tree $T$, $C_u \cap C_v = \emptyset$.
\end{itemize}
\end{observation}

\noindent
Let $\gamma_v(C) = \abs{C \cap E_v} = \abs{C_v \cap E}$. Thus, 
$\abs{C} = \sum_{v \in V(T)} \gamma_v(C)$.
For any subtree $T'$, we define 
$$\gamma_{T'}(C) = \sum_{v \in V(T')} \gamma_{v}(C).$$

\subsection{Partitioning the tree \texorpdfstring{$T$}{T} into smaller weight subtrees}
\label{sec:partitioning}
We classify the circuits $C$ of $M$ according to the distribution of $\gamma_v(C)$ for $v \in V(T)$.
To do this, we create a signature for each circuit.
Let $C$ be a circuit of $M$ of size at most $\alpha r$.
For circuit $C$, we first find a center vertex of the tree $T$.
That is, a vertex $u_0$ in the tree $T$ such that if $T_1,T_2,\dots, T_k$ are the subtrees
obtained by deleting $u_0$, then $\gamma_{T_i}(C) \leq  \alpha r/2$ for each $i \in [k]$.  
The following claim shows that such a vertex exists.

\begin{claim}[Balanced partition of a tree]
\label{cla:center}
Let $T $ be a tree and let each vertex $v \in V(T)$ be associated with a non-negative integer $\gamma_v$.
Let $q := \sum_{v \in V(T)} \gamma_{v } $. 
Then there exists a center vertex $c(T) \in V(T)$ such that 
for the subtrees $T_1,T_2,\dots, T_k$ obtained by deleting $c(T)$ from $T$,
we have 
$$\sum_{v \in T_i} \gamma_v \leq q/2$$ for each $1\leq i \leq k $.
\end{claim}
\begin{proof}
We give a procedure to find the vertex $c(T)$.
Start with an arbitrary vertex $v_0$. 
Let $T_1, T_2,\dots, T_k$ be the subtrees obtained by removing $v_0$.
Let $T_i$ be the subtree which maximizes $\sum_{v \in T_i} \gamma_{v}$
and let $q_0$ be the maximum value of this sum.
If $q_0 \leq q/2$, then we output $v_0$ as the desired vertex. 
On the other hand, if we have $q_0 > q/2$ 
then we get that 
$$\sum_{v \in V(T) \setminus V(T_i)} \gamma_v  < q/2. $$
Let $v_1$ be the vertex in $T_i$ connected to $v_0$.
If $v_1$ is a leaf, that is, $T_i$ is a single vertex $v_1$
 then we output $v_1$ as the desired vertex $c(T)$. 
The vertex $v_1$ has the desired property because its removal gives a single subtree 
that is $T-T_i$, which has weight $< q/2$.

In the other case, when $v_1$ is not a leaf, we repeat the procedure.
That is, check if there is a subtree connected to $v_1$ with weight $>q/2$
and if yes, then move to the neighbor of $v_1$ in this subtree. 
We claim that this algorithm terminates since we never go back to a previous vertex.
To see this, observe that the subtree connected to $v_1$ which contains $v_0$ is $T-T_i$ and its weight is less than $q/2$. 
Thus, we do not move back to $v_0$.
\end{proof}

\noindent
Note that the center vertex is not uniquely defined. For our purposes, any vertex with the desired property would work.
We apply the center-finding procedure recursively on the subtrees $T_1, T_2, \dots, T_k$,
till we reach a point where for any of the obtained subtrees $T'$, the sum $\gamma_{T'}$ is at most $r/2$.

\begin{claim}[Balanced recursive partition of a tree]
\label{cla:centers}
Let $T $ be a tree with each vertex $v$ associated with a non-negative integer $\gamma_v$.
 Let 
$$\gamma_T= \sum_{v \in T} \gamma_v \leq 2^p r$$ 
for some integers $p,r \geq 0$.
Then there exists a set $U$ of at most $ 2^{p+2}$ vertices such that 
on removing these vertices from the tree $T$, for each obtained subtree $T'$  we have that
$$\gamma_{T'} = \sum_{v \in T'} \gamma_v  \leq r/2.$$
\end{claim}
\begin{proof}
We construct the set $U$ in $p+1$ rounds. 
Let $U_i \subseteq V(T)$ be the set of vertices we construct after the $i$-th round and $\T_i$ be the set of 
subtrees obtained by deleting the vertices in $U_i$ from the tree $T$.
We want to ensure that for each subtree $T' \in \T_i$, 
we get
$$ \gamma_{T'}  \leq 2^{p-i} r.$$
Note that this is true for $i=0$ with the set $U_0 = \emptyset$ and $\T_0 = \{T\}$.
For $0\leq i \leq p$, we describe the $(i+1)$-th round
assuming that  we have $U_i$ and $\T_i$ with the desired properties.
Among the subtrees in $\T_i$, we need not care about the subtrees $T'$
that already have $\gamma_{T'} \leq 2^{p-i-1}r$. 
So, consider the subtrees $T' \in \T_i$ that have
$$ \gamma_{T'} > 2^{p-i-1}r.$$
There can be at most $2^{i+1}$ such subtrees, as $\gamma_T \leq 2^p r$.
For each such subtree $T'$ we find a center vertex $c(T')$ using \cref{cla:center}.
This choice of $c(T')$ is such that removing it from $T'$ gives us subtrees $T''$ having
$$ \gamma_{T''}  \leq 2^{p-i-1} r.$$
Thus,  we define the set $U_{i+1}$ to be the union of $U_i$ and the set of centers $c(T')$, i.e., 
$$U_{i+1} =U_i  \cup \{c(T') \mid T' \in \T_i, \; \gamma_{T'} > 2^{p-i-1}r \} .$$
By construction, $U_{i+1}$ has the desired property.
For the size bound, observe that
$$\abs{U_{i+1}} \leq \abs{U_{i}} + 2^{i+1} .$$
Solving the recurrence, we get 
$$\abs{U_{p+1}} \leq 2^{p+2}.$$
\end{proof}

\subsection{Classifying the circuits using signatures}
\label{sec:classifying}
Recall that $C$ is a circuit of $M$ of size at most $\alpha r$, which implies $\gamma_T(C) \leq \alpha r$.
From \cref{cla:centers},
there exists a set $U(C) \subseteq V(T)$ of at most $4\alpha$ vertices whose removal  
ensures that for any obtained subtree $T'$ 
we have 
\begin{equation}
\label{eq:tree-rby2}
\gamma_{T'}(C)  \leq r/2.
\end{equation}
To bound the number of such circuits, we classify them by signature. Then, we bound the number of circuits
associated with a fixed signature and also the number of possible signatures.
 $\signature(C)$ will be defined for any circuit $C$ of $M$ with $\abs{C} \leq \alpha r$,
based on a set $U(C)$ which ensures (\ref{eq:tree-rby2}).
Note that the set $U(C)$, and thus the signature of $C$, will not be uniquely defined. 
We take any arbitrary choice of $U(C)$ for the signature.

Recall that $C = \triangle_{v \in V(T)} C_v$.
Let $u$ be a vertex in $U(C)$ and let $N(u)$ be the set of its neighboring vertices in the tree $T$.
We define a subset $N^*(u)$ of neighboring vertices $N(u)$ as follows:
\begin{align*}
N^*(u) := \{v \in N(u) \mid & \; v \text{ is on a path connecting } u \text{ to another vertex } u' \in U(C) \\
& \text{ or } v \in U(C) \}.
\end{align*}
We claim that the set $C_u$ and the set $C_v$ for any vertex $v \in N^*(u)$ must be non-empty. 
This is because (i) from the construction of the set $U(C)$,   if a subtree $T'$ of $T$
contains a vertex from $U(C)$, then it must be the case that $\gamma_{T'}(C) \neq 0$, i.e., $C_{T'}$ is non-empty
and (ii) the vertices $v$ with $C_v \neq \emptyset$ form a connected subgraph of $T$ (\cref{obs:CuCv}).

Recall from \cref{obs:CuCv} that  $C_u$ and $C_v$ have a single common element coming from $E_u \cap E_v$. 
 For each $u \in U(C)$ and for each $v \in N^*(u)$, we put this common element of $C_u$ and $C_v$ in $\signature(C)$.
This finishes the construction of the signature. 
$$\signature(C) := 
 \{(u, E_u \cap C_v) \mid u \in U(C), \; v \in N^*(u)\}.$$
First, we upper bound the number of possibilities for the signature. 
\begin{claim}
\label{cla:num-signatures}
There are at most $(9m)^{4\alpha}$ possibilities for  $\signature(C)$ with $\abs{C} \leq \alpha r$. 
\end{claim}
\begin{proof}
There are at most $4\alpha$ vertices in $U(C)$ and thus, there are at most $m^{4 \alpha}$ possibilities for $U(C)$.
Next, we would like to bound the number of tuples $(u, E_u \cap C_v)$ in $\signature(C)$.
Observe that this number is bounded by 
$$\sum_{u \in U(C)} \abs{N^*(u)}.$$
We claim that this sum is bounded by $2\abs{U(C)} \leq 8 \alpha$. 
To see this, make the tree $T$ rooted at an arbitrary vertex. 
For any vertex $u \in U(C)$, let $A(u) \in U(C)$ be its nearest ancestor in the rooted tree $T$.
For each $u \in U(C)$, we count two vertices -- first, the neighbor of $u$ connecting it to $A(u)$ and second, the neighbor of $A(u)$ connecting it to $u$.
This way, we have counted every vertex in $\cup_{u \in U(C)} N^*(u)$.

Recall that the element $E_u \cap C_v$ comes from $E_u \cap E_v$, which contains at most $3$ elements. 
So, for each  tuple $(u,C_u \cap E_v)$, there are at most $3$ possibilities.
Thus, for a given choice of $U(C)$, there are at most $3^{8 \alpha}$ possibilities for $\signature(C)$. 
By multiplying this number with the number of possibilities for the set $U(C)$,  
we get an upper bound on the number of possible signatures. 
\end{proof}

\subsection{Bounding the number of circuits for a given signature}
\label{sec:givenSignature}
Our aim is to bound the number of circuits with a given signature.
 Let $\S$ be a signature of a circuit of size at most $\alpha r$.
Let $\C_{\S}$ be the set of all circuits $C$ of size at most $\alpha r$ such that 
$\signature(C) = \S$.
Let $U = \{u : (u,e) \in \S \}$ be the set of vertices of the tree T appearing in the signature $\S$.
Let $\T$ be the set of subtrees obtained after removing the vertices in $U$ from $T$.
Then, for the ground set $E$ of matroid $M$ we can write
$$ E = \left( \triangle_{u \in U} E_u \right) \triangle \left( \triangle_{T' \in \T} E_{T'} \right).$$
Also, every circuit $C$ of $M$ can be written as
$$ C = \left( \triangle_{u \in U} C_u \right) \triangle \left( \triangle_{T' \in \T} C_{T'} \right).$$
To bound the number of circuits in $C$ in $\C_\S$, we  bound the number of possibilities for the component circuits $C_u$ and $C_{T'}$
for $u \in U$ and $T' \in \T$.
We know that for any subtree $T' \in \T$ and for any circuit $C \in \C_\S$, 
$$\gamma_{T'}(C) = \abs{E_{T'} \cap C} \leq r/2.$$

\noindent
For any subtree $T' \in \T$, let $E_{T-T'}$ be the set $\triangle_{v \in V(T) \setminus V(T')} E_v$.
Then we have $E = E_{T'} \triangle E_{T-T'}$.
 The circuit $C_{T'}$ also has elements from $E_{T'} \cap E_{T-T'}$, which we call its boundary elements. 
We  first show that for all circuits $C \in \C_\S$, there is a unique possibility of the circuit $C_{T'}$, once we fix these boundary elements
$C_{T' } \cap E_{T-T'}$.
\begin{claim}
Let $C$ and $D$ be two circuits in $\C_\S$ and $T' \in \T$.
If $C_{T'} \cap E_{T-T'} = D_{T'} \cap E_{T-T'}$, then
$C_{T'} = D_{T'}$.
\label{cla:unique}
\end{claim}
\begin{proof}
Assume, for the sake of contradiction, that $C_{T'} \ neq D_{T'}$.
Consider the symmetric difference $C_{T'} \triangle D_{T'}$, which is a disjoint union of circuits in $M_{T'}$ (\cref{fac:symm-diff-circuits}).
From the assumption of the claim,  the set $C_{T'} \triangle D_{T'}$
does not have any elements from $E_{T-T'}$.
In particular, this means that $C_{T'} \triangle D_{T'}$
 is actually a subset of $E  = E_{T'} \triangle E_{T-T'}$.
In this case, the set $C_{T'} \triangle D_{T'}$ is, in fact, a disjoint union of circuits in $M$ (\cref{fac:circuit-in-one}).
 However, $$\abs{C_{T'} \triangle D_{T'} } \leq \abs{C \cap E_{T'}} + \abs{D \cap E_{T'}} \leq r/2 +r/2 =r,$$ 
 which contradicts the hypothesis of the theorem that $M$ has no circuits of size at most $r$.
\end{proof}

\noindent
We argue that for some of the subtrees $T' \in \T$, the boundary elements of circuit $C_{T'}$ indeed get 
fixed by the signature $\S$.
We distinguish two types of subtrees in the set $\T$. 
Let $\T_0 \subseteq \T$ be the set of all the subtrees which have at most one vertex of $U$ adjacent to them.
The other set $\T_1 = \T  \setminus \T_0$ contains all subtrees with at least two vertices of $U$ adjacent to them.

We first show that for any subtree $T' \in \T_1$,
all circuits $C_{T'}$, for $C \in \C_\S$, have the same boundary elements.

\begin{claim}
\label{cla:T1-signature}
For any two circuits $C,D \in \C_\S$ and any subtree $T' \in \T_1$,
$$C_{T'} \cap E_{T-T'} = D_{T'} \cap E_{T-T'}.$$
\end{claim}
\begin{proof}
Let $u$ be any vertex in $U$ that is adjacent to $T'$ and let $v \in V(T')$ be the vertex in $T'$
connected to $u$.
 Since $T'$ is adjacent to at least one more vertex $u' \in U$ that is different from $u$, 
 the vertex $v$ must be on the path connecting $u$ to $u'$.
In other words, $v \in N^*(u)$.
Thus, the tuple $(u, E_u  \cap C_v)$ is in $\signature(C)$ for any circuit $C \in \C_\S$.
Since $C$ and $D$ have the same signature, we get  $E_u \cap C_v = E_u \cap D_v$. 
Since this is true for any vertex $u \in U$ that is adjacent to $T'$ and its neighbor $v$ in $T'$, 
we can write $$E_{T-T'} \cap C_{T'} = E_{T-T'} \cap D_{T'}.$$
\end{proof}

\noindent
Claims~\ref{cla:unique} and \ref{cla:T1-signature} together imply that 
for any subtree $T' \in \T_1$, there is only one possible circuit $C_{T'}$ for all $C \in \C_\S$.
\begin{lemma}[Unique projection on a subtree in $\T_1$]
\label{cla:T1}
For any subtree $T' \in \T_1$,
the set 
$\{C_{T'} \mid C \in \C_\S \}$
has cardinality exactly one.
\end{lemma}

\noindent
We move on to the subtrees in $\T_0$.
A claim similar to \cref{cla:T1-signature} cannot be made about a subtree $T' \in \T_0$.
This is because the signature $\S$ does not fix the boundary elements of $C_{T'}$ for $T' \in \T_0$.

Recall that by definition any subtree $T' \in \T_0$ has a unique vertex of $U$ adjacent to it, say $u$. 
This means, $E_{T'} \cap E_{T-T'} = E_{T'} \cap E_u$, which has at most three elements in it.
Thus, the number of possibilities for the set $C_{T'} \cap E_{T-T'}$ is at most three.
From \cref{cla:unique}, we get that there are at most three distinct possibilities of $C_{T'}$ for $C \in \C_\S$.
However, there can be a large number of subtrees in $\T_0$, possibly $O(m)$. 
This would mean the total number of possibilities for circuits $C_{T'}$ for all $T' \in \T_0$ can be $3^{O(m)}$.
Instead, we use a different strategy where we bound the number of possibilities for circuits $C_u$ for $u \in U$ and $C_{T'} $ for $T' \in \T_0$
simultaneously. 

For any vertex $u \in U$, let $\T_{0,u} \subseteq \T_0$ be the set of subtrees in $\T_0$ that have a vertex adjacent to $u$.
\begin{lemma}
\label{lem:numberCu}
For any vertex $u \in U$, the cardinality of the set 
$$\{ C_u \triangle \left( \triangle_{T' \in \T_{0,u}} C_{T'} \right) \mid C \in \C_\S \}$$
 is at most $(2m)^{10\alpha}$.
\end{lemma}
\begin{proof}
The idea is to design a one-to-one map from the set in the lemma to the set of circuits of 
the matroid $M_u$, which is size-preserving in some sense. 
Then we will bound the number of circuits in $M_u$ of a certain size using \cref{thm:graphic-cographic} when
it is graphic or cographic and trivially when it is the $R_{10}$ matroid.
Our map is simply 
$$ \phi \colon C_u \triangle \left( \triangle_{T' \in \T_{0,u}} C_{T'} \right) \mapsto C_u.$$
First, we argue that the map $\phi$ is one-to-one.
For a given circuit $C_u$, we show that it has a unique pre-image under $\phi$.
For any subtree  $T' \in \T_{0,u}$,
let $S_{T'} := E_{T'} \cap E_u$ be the set of common elements between $E_{T'}$ and $E_u$.

If $C_u \cap S_{T'} = \emptyset$ for some $T' \in \T_{0,u}$, then $C_{T'}$ must be empty (\cref{obs:CuCv}).
In this case, there is a unique choice for $C_{T'}$.

Consider the case when $C_u \cap S_{T'} = \{s\}$ for some $T' \in \T_{0,u}$. 
Let $D \in \C_\S$ be another circuit such that  $C_u = D_u$.
Since $C_u \cap S_{T'} = \{s\}$, we also have $C_{T'} \cap E_u = \{s\}$ (\cref{obs:CuCv}).
Similarly, $D_{T'} \cap E_u = \{s\}$, which means $$D_{T'} \cap E_u = C_{T'} \cap E_u.$$
Recall that since $u$ is the single vertex in $T-T'$ adjacent to $T'$, we have $E_{T'} \cap E_u = E_{T'} \cap E_{T-T'}$.
Thus, we get that $$D_{T'} \cap E_{T-T'} = C_{T'} \cap E_{T-T'}.$$
Hence, from \cref{cla:unique}, $D_{T'} = C_{T'}$.
This proves that given a $C_u$, there is a unique choice for $C_{T'}$ for each $T' \in \T_{0,u}$.

In fact, we can conclude something stronger from the above argument. 
Suppose for a subtree $T' \in \T_{0,u}$ and element $s \in S_{T'}$, there exists a circuit $C \in \C_\S$ 
with $C_u \cap {S_T'} = \{s\}$.
Then for every circuit $C \in \C_\S$ with $C_u \cap {S_T'} = \{s\}$, 
there is a unique choice for $C_{T}'$.
Let $D^s_{T'}$ denote this unique choice for $T' \in \T_{0,u}$ and  $s \in S_{T'}$.
 
Before we bound the number of circuits in $M_u$, let
 us delete a few elements from $M_u$ that will not affect the range of the map $\phi$. 
These are elements $s \in S_{T'}$ for  $T' \in \T_{0,u}$ such that 
there is no circuit $C \in \C_\S$ 
with $C_u \cap {S_T'} = \{s\}$.
Let $\overline{M}_u$ be the matroid obtained from $M_u$ after deleting these elements.
Note that $\oM_u$ remains a graphic, cographic or the $R_{10}$ matroid (\cref{fac:closed}). 
Clearly, $\phi$ is also a one-to-one map from the set in the lemma to the set of circuits of the matroid $\overline{M}_u$.
For later use, note that when $M_u$ is cographic, any 3-element circuit which is the degree-3-star in the representing graph retains the degree-3-star representation in $\overline{M}_u$, if it survives.
This is because deleting an element in $M_u$ corresponds to contracting an edge outside the degree-3-star in the representing graph.

Now, we would also like $\phi$ to be size-preserving in some sense.
For this reason, we have to assign some integer weights to elements in $\overline{M}_u$ and then consider 
the weight of a circuit in $\overline{M}_u$. 
Recall that $N^*(u)$ is the set of all the neighboring vertices of $u$ in the tree $T$ which 
are not part of any $T' \in \T_{0,u}$.
For any $v \in N^*(u)$, let $S_{v} := E_{v} \cap E_u$. 
We define the weights for the elements in $\overline{M}_u$ as follows:
\[
w(e) =
\begin{cases}
 \abs{D^e_{T'}} - 1 & \text{ if } e \in S_{T'}  \text{ for some } T' \in \T_{0,u} \\
 r & \text{ if } e \in S_v \text{ for some } v \in N^*(u) ,  \\
 1 & \text{ otherwise.}
\end{cases}   
\]
Recall that for any circuit $C_u$ of $\overline{M}_u$, its weight is defined to be $w(C_u) = \sum_{e \in C_u} w(e)$.
\begin{claim}
\label{cla:weight}
For any $C \in \C_\S$, $ w(C_u) \leq \abs{C_u \triangle \left( \triangle_{T' \in \T_{0,u}} C_{T'} \right)} +  4 \alpha r$.
\end{claim}
\begin{proof}
From the arguments above, 
for a subtree $T' \in \T_{0,u}$, if $C_u \cap S_{T'} = \{s\}$ then 
there  $C_{T'} = D^s_{T'}$  and if 
$C_u \cap S_{T'} = \emptyset$ then $C_{T'} = \emptyset$. 
From the construction of the signature, we know that $C_u$ contains exactly one element from $S_v$ 
for every $v \in N^*(u)$. And $\abs{N^*(u)} \leq \abs{U} \leq 4\alpha$.
Using these two facts, the claim follows directly from the definition of the weight function $w(\cdot)$.
\end{proof}

\noindent
Note that since $\abs{C}$ is at most $\alpha r$ for any $C \in \C_\S$, 
the size of the set $C_u \triangle ( \triangle_{T' \in \T_{0,u}} C_{T'} )$ is also at most $\alpha r$.
And thus from \cref{cla:weight}, the weight $w(C_u)$ is at most $5 \alpha r$ for any $C \in \C_\S$.
To bound the cardinality of the set in the lemma, 
it will be sufficient to bound the number of circuits $D_u$ of $\overline{M}_u$ with weight $w(D_u) \leq 5\alpha r$ and $|D_u \cap S_{T'}|\leq 1$ for each $T' \in \T_{0,u}$.
If $\overline{M}_u$ is the $R_{10}$ matroid, then its number of circuits is trivially bounded by $2^{10}$.
Otherwise $\overline{M}_u$ is a graphic or cographic matroid, and we can use \cref{thm:graphic-cographic}.
 For that we will need a lower bound on the weight of any circuit in $\overline{M}_u$.

\begin{claim}
\label{cla:no-circuit}
There is no circuit $D_u$ of $\overline{M}_u$ with weight $w(D_u) \leq r$,
except possibly 3-element circuits $S_{T'}=E_{T'} \cap E_u$ for some $T'\in \T_{0,u}$.
\end{claim}
\begin{proof}
For the sake of contradiction, let there be such a circuit $D_u$, which is not a 3-element circuit $S_{T'}$ for some $T'\in \T_{0,u}$.
$D_u$ will also be a circuit of $M_u$.
We show that this implies the existence of a circuit of $M$ of size at most $r$, which would be a contradiction.

For any subtree $T' \in \T_{0,u}$, we define $D_{T'}$ as follows.
\begin{itemize}
    \item If $D_u \cap S_{T'} = \emptyset$, then define $D_{T'} = \emptyset$.
    \item If $D_u \cap S_{T'} = \{s\} $, then define $D_{T'} = D^s_{T'}$.
Recall that there could be elements $s \in S_{T'}$, for which $D^s_{T'}$ would not be defined,
but such elements were deleted while defining $\overline{M}_u$.
\item  If  $D_u \cap S_{T'} = \{s_1,s_2\} $, then define $D_{T'} = D^{s_1}_{T'}\triangle D^{s_2}_{T'}$.

\item  If  $D_u \cap S_{T'} = \{s_1,s_2,s_3\} $, then $D_u$ must be same as $S_{T'} = \{s_1,s_2,s_3\}$, because  
$S_{T'}$ is a circuit. 
We ruled out this case at the beginning.
\end{itemize}
Consider the set 
$$ D =  D_u \triangle \left( \triangle_{T' \in \T_{0,u}} D_{T'} \right).$$
Note that $D_u$ cannot have any elements from $S_v$ for any $v \in N^*(u)$, otherwise its weight will be strictly larger than $r$.
Thus, from the definition of the weight function $w(\cdot)$, it is clear that 
 $\abs{D} \leq w(D_u)    \leq r$.
We also have $D\neq\emptyset$.
If $D_{T'}$ is empty for each $T'\in \T_{0,u}$, then $D=D_u\neq \emptyset$.
Otherwise consider a nonempty $D_{T'}$ for some $T'\in \T_{0,u}$.
$D_{T'}$  is either $D^s_{T'}$ or $D^{s_1}_{T'} \triangle D^{s_2}_{T'}$. 
In either case, $D_{T'}$ must have an element of $E_{T'}$ outside $S_{T'}$ to contain a circuit. 
This element cannot be canceled by any element from another subtree.

We argue that $D$ is, in fact, a disjoint union of circuits of $M$, which contradicts the fact that $M$ has no circuit of size at most $r$.
Let $T''$ be the subtree obtained by joining $u$ and all the subtrees in $\T_{0,u}$.
By construction, $D$ is a disjoint union of circuits of the matroid $M_{T''}$ (\cref{fac:disjointCircuits}).
We claim that  $D$ has no elements from $E_{T''} \cap E_{T-T''}$, 
which implies that $D$ is a disjoint union of circuits of $M$ (\cref{fac:circuit-in-one}).
To see this claim, observe that the set of common elements $E_{T''} \cap E_{T-T''}$ 
is the set $\cup_{v \in N^*(u)} S_{v}$. 
As mentioned above, $D_u$ cannot have any elements from $\cup_{v \in N^*(u)} S_{v}$. 
\end{proof}

\noindent
Now, before we can apply 
\cref{thm:graphic-cographic} on the matroid $\oM_u$,
we need to get rid of any possible 3-element circuit $S_{T'}$ for some $T'\in \T_{0,u}$, whose weight can possibly be $r$ or less.
We do this by the standard $\Delta$-$Y$ transform. 
Let $\mathcal F$ be the set of subtrees $T'\in\T_{0,u}$
for which $S_{T'}$ has three elements and all three survive
in $\overline{M}_u$.
Each such $S_{T'}$ remains a circuit of $\overline{M}_u$.
Recall that the gluing sets $S_{T'}$ for $T' \in \mathcal F$ are pairwise disjoint.
We transform \emph{every} gluing set for subtrees in $\mathcal F$,
whether or not its weight is at most $r$.
For each $T'\in\mathcal F$, replace the 3-element circuit
$S_{T'}$ by a 3-element cocircuit $S'_{T'}$ as follows.
\begin{itemize}
    \item If $\oM_u$ is graphic, then $S_{T'}$
    would be a 3-cycle in the representing graph, with edges, say, $s_1=(a_2,a_3), s_2=(a_1,a_3), s_3=(a_1,a_2)$. 
    We remove the three edges, add a new vertex $a_0$ and edges $s_1'=(a_0,a_1)$, $s_2'=(a_0,a_2)$,
    $s_3'=(a_0,a_3)$.
    \item If $\oM_u$ is cographic, then $S_{T'}$
    would be the edges incident on a degree-3 vertex (from \cref{lem:uDecomposition}) in
    the representing graph,
    with edges, say
    $s_1=(a_0,a_1)$, $s_2=(a_0,a_2)$,
    $s_3=(a_0,a_3)$.
    We remove the three edges and $a_0$, and add edges
     $s_1'=(a_2,a_3), s_2'=(a_1,a_3), s_3'=(a_1,a_2)$.
\end{itemize}
Let us call the new matroid $\oM_u'$.
We define the weights of the new elements as follows. 
\begin{eqnarray*}
    w(s_1') &=& (w(s_2)+w(s_3)-w(s_1))/2\\
    w(s_2') &=& (w(s_1)+w(s_3)-w(s_2))/2\\
    w(s_3') &=& (w(s_1)+w(s_2)-w(s_3))/2\\
\end{eqnarray*}
The following claim implies that the weights defined above are non-negative.

\begin{claim}

For every $T'\in\mathcal F$ with
$S_{T'}=\{s_1,s_2,s_3\}$,
the weights $w(s_1),w(s_2),w(s_3)$
satisfy the triangle inequalities.

\end{claim}

\begin{proof}
Fix such an gluing set $S_{T'}=\{s_1,s_2,s_3\}$, and put
\[
a_i:=w(s_i)=|D^{s_i}_{T'}|-1.
\]
The bound on the projections of circuits in $\C_\S$
gives $a_i\leq r/2$ for each $i$.
We show that $a_1\leq a_2+a_3$;
the other two inequalities are symmetric.
The set
\[
D^{s_2}_{T'}\triangle D^{s_3}_{T'}
\triangle S_{T'},
\]
which is a union of circuits,
intersects $S_{T'}$ exactly in $\{s_1\}$,
and hence contains a circuit $Q$ of $M_{T'}$
with $Q\cap S_{T'}=\{s_1\}$.
Moreover,
\[
|Q|-1\leq a_2+a_3,
\]
since $Q\setminus \{s_1\}$ is contained in
$D^{s_2}_{T'}\cup D^{s_3}_{T'}\setminus \{s_2,s_3\}$.

If $Q=D^{s_1}_{T'}$, then
\[
a_1=|Q|-1\leq a_2+a_3,
\]
as desired.
Otherwise, $Q\triangle D^{s_1}_{T'}$ is a nonempty
disjoint union of circuits of $M$:
the two circuits have the same  gluing element
element $s_1$, which cancels.
As $M$ has no circuit of size at most $r$, we have
\[
r<|Q\triangle D^{s_1}_{T'}|
\leq (|Q|-1)+a_1
\leq a_2+a_3+a_1.
\]
Consequently,
\[
a_2+a_3>r-a_1\geq r/2\geq a_1.
\]
This proves the triangle inequality.
\end{proof}
\noindent
Now, we relate the circuits of the matroids $\oM_u$ and $\oM_u'$. 

\begin{claim}
There is a weight-preserving bijection between the circuits
$D_u$ of $\oM_u$ satisfying
\[
|D_u\cap S_{T'}|\leq 1
\qquad\text{for every }T'\in\mathcal F,
\]
and all circuits $D'_u$ of $\oM'_u$.

\label{cla:circuitsDu'}
\end{claim}

\begin{proof}
Outside the modified gluing sets,
 $D_u$ and $D'_u$
have the same elements.
For each $T'\in\mathcal F$, write
\[
S_{T'}=\{s_1,s_2,s_3\},
\qquad
S'_{T'}=\{s'_1,s'_2,s'_3\},
\]
and set
\[
D'_u\cap S'_{T'}=
\begin{cases}
S'_{T'}\setminus\{s'_i\},
  &\text{if }D_u\cap S_{T'}=\{s_i\},\\
\emptyset,
  &\text{if }D_u\cap S_{T'}=\emptyset.
\end{cases}
\]
For a graphic component, the $\Delta$--$Y$ exchange
replaces each use of one triangle edge by the
two-edge path through the new degree-three vertex
with the same endpoints.
A simple cycle in the modified graph uses either
zero or two of these star edges, so contracting
that two-edge path back to the corresponding
triangle edge gives the inverse map on circuits.
For a cographic component, circuits are cuts of
the underlying graph.
The inverse $Y$--$\Delta$ exchange replaces a
degree-three star by a triangle.
A cut in the new graph intersects this triangle
in either zero or two edges; assigning the removed
degree-three vertex to the side containing the
two corresponding neighbors gives the inverse map on cuts.
Connectivity of both shores is preserved,
so this is again a bijection.
As the gluing sets for subtrees in $\mathcal F$ are
disjoint, these local exchanges can be
applied successively; in the cographic case
their degree-three centers cannot be adjacent,
so the remaining stars are preserved during the exchanges.
Finally, for $\{i,j,k\}=\{1,2,3\}$,
\[
w(s'_j)+w(s'_k)=w(s_i).
\]
Hence each exchange preserves the total weight,
and all untouched elements retain their
original weights.
\end{proof}

\noindent
Every circuit of $\oM'_u$ of weight at most $r$
would therefore correspond to a circuit of $\oM_u$
of the same weight, and cannot be equal to $S_{T'}$
for any $T' \in \T_{0,u}$.
Thus $\oM'_u$ has no circuit of weight at most $r$.

To get the desired bound, it suffices to bound the number of circuits of $\oM_u'$ with weight at most $5\alpha r$.
Using \cref{thm:graphic-cographic} 
we get that this number
is bounded by $ (2m)^{10\alpha}$.
\end{proof}

\noindent
Using \cref{cla:T1} and \cref{lem:numberCu} we can bound the number of circuits in $\C_\S$.
\begin{lemma}
\label{lem:circuits-S}
For any signature $\S$, the number of circuits in the set $\C_\S$ is at most $m^{O(\alpha^2)}$.
\end{lemma}
\begin{proof}
For any circuit $C \in \C_\S$, we write it as 
$$ \left( \triangle_{T' \in \T_1} C_{T'} \right) \triangle \left(  \triangle_{u \in U} \left( C_u \triangle \left( \triangle_{T' \in \T_{0,u}} C_{T'} \right) \right) \right).$$
From \cref{cla:T1}, the number of possibilities for the set $\triangle_{T' \in \T_1} C_{T'}$ is  $1$.
From \cref{lem:numberCu}, for each $u \in U$ the number of possibilities for the set $ C_u \triangle \left( \triangle_{T' \in \T_{0,u}} C_{T'} \right)$ 
is at most $ (2m)^{10\alpha}$.
The number of vertices in $U$ is at most $4 \alpha$.
Together this gives us a bound of $( (2m)^{10\alpha})^{4 \alpha}$ on the number of circuits in $\C_\S$.
\end{proof}

\noindent
By \cref{cla:num-signatures}, the number of different possible signatures for a circuit in $M$ of size at most $\alpha r$
is $m^{O(\alpha)}$.
Together with \cref{lem:circuits-S}, we get that the number of circuits in $M$ of size at most $\alpha r$ 
is $m^{O(\alpha^2)}$.
This finishes the proof of \cref{thm:num-circuits}.

\section{Number of \texorpdfstring{$\alpha$}{alpha}-shortest vectors in a totally unimodular lattice}
\label{sec:lattice}

In this section, we show how \cref{thm:lattice} follows from 
\cref{thm:num-circuits}.
Recall that for an $n\times m$ matrix $A$, the lattice $L(A)$ is defined as:
$$L(A ) = \{v \in \Z^m \mid A v =0\}.$$
We first define circuits of a matrix $A$, which are vectors in $L(A)$, and show a correspondence between the 
circuits of a TU matrix $A$ and the circuits of the associated regular matroid $M(A)$.
Thus, an upper bound on the number of near-minimum circuits in $M(A)$ (\cref{thm:num-circuits}) implies an upper bound on the number of near-minimum circuits of 
the matrix $A$.
Finally, we argue that any $\alpha$-minimum vector in $L(A)$ comes from a combination of at most $\alpha^2$-many 
$\alpha$-minimum circuits of $A$.
A specialized version  of this statement was shown in \cite{GTV18} -- 
any $2$-minimum vector in $L(A)$ comes from a
$2$-minimum circuit of $A$.

\begin{definition}[\textbf{Circuits of a matrix}]
For a matrix~$A$, a vector $u \in L(A)$ is a circuit of $A$ if 
for any vector $v \in L(A)$ with
 $\supp(v) \subseteq \supp(u)$, it must be that 
$ v = \gamma u$ for some integer $\gamma$.
\end{definition}

\noindent
Note that the circuits of $A$ come in pairs, in the sense that if $u$ is a circuit of~$A$, then so is $-u$.
For any rational matrix $A$, any linear dependence among its columns can always be taken over integers.
Thus by definition, for a circuit $u$ of $A$,
the columns in~$A$ corresponding to the set~$\supp(u)$
are minimally dependent. 
The following observation follows immediately.
\begin{observation}  
\label{obs:circuits}
Let $A$ be a rational matrix and  $M$ be the matroid represented by it.
Then the circuits of $M$ have a one-to-one correspondence 
with the circuits of~$A$ (up to change of sign).
\end{observation}

\noindent
\cref{obs:circuits} together with \cref{thm:num-circuits} gives the following corollary.
Let $\norm{\cdot}$ denote the $\ell_2$-norm of a vector.

\begin{corollary}
\label{cor:TU-circuits}
Let $A$ be an $n\times m$ TU matrix. 
If for every circuit $u$ of $A$, we have $\norm{u} >r$ then the number of its circuits $u$ with $\norm{u} \leq \alpha r$
is $m^{O(\alpha^4)}$.
\end{corollary}
\begin{proof}
It is known that every circuit of a TU matrix 
has its coordinates in $\{-1,0,1\}$ (see~\cite[Lemma 3.18]{Onn10}).
So, if $u$ is a circuit of $A$ with $\supp(u) = \gamma$ then $\norm{u} = \sqrt{\gamma}$.
Thus, any circuit $u$ of $A$ with $\norm{u} \leq \alpha r$ corresponds to a circuit of the regular matroid $M(A)$ of size at most $\alpha^2 r^2$.
The desired bound follows from \cref{thm:num-circuits}.
\end{proof}

\noindent
We now show that we can get an upper bound on the number of all short vectors in $L(A)$ from \cref{cor:TU-circuits}.
We define a notion of conformality among two vectors and show that every vector in $L(A)$ is a conformal combination of circuits of $A$.

\begin{definition}[\textbf{Conformal vectors \cite{Onn10}}]
Let $u,v \in \R^m$. 
We say that~$u$ is \emph{conformal to}~$v$, denoted by
$u \sqsubseteq v $, 
if $u_iv_i \geq 0$ and $\abs{u_i} \leq \abs{v_i}$, for each $1\leq i \leq m$.
\end{definition}

\noindent
The following lemma, which says that each vector in $L(A)$ is a conformal sum of circuits, follows from \cite[Lemma 3.2 and  3.19]{Onn10}.

\begin{lemma}
\label{lem:conformal-sum}
Let $A$ be a TU matrix. 
Then for any nonzero vector $v \in L(A)$, we have 
$$v = u_1 + u_2 + \cdots + u_p$$ 
where each $u_i$ is 
a circuit of~$A$ and is conformal to~$v$.
\end{lemma}

\noindent
We have the following easy observation for a conformal sum.

\begin{observation}
\label{obs:conformal}
If $v,u_1,u_2, \dots, u_p \in \Z^m $ are such that $v = u_1 +u_2 + \cdots + u_p$ and each $u_i$
is conformal to $v$, then 
$$\norm{v}^2 \geq \norm{u_1}^2 + \norm{u_2}^2 + \cdots + \norm{u_p}^2.$$
\end{observation}

\noindent
We are ready to prove \cref{thm:lattice}.

\begin{proof}[Proof of \cref{thm:lattice}]
From the assumption in the theorem, 
for any circuit $u$ of $A$, we have $\norm{u} > \lambda$.
Consider a vector $v \in L(A)$ with $\norm{v} \leq \alpha \lambda$.
From \cref{lem:conformal-sum}, we can write $v$ as a conformal sum of circuits 
$$v = u_1 + u_2 + \dots + u_p.$$
We know that  $\norm{u_i} > \lambda$ for each $i$.
This together with \cref{obs:conformal} implies that 
$$(\alpha \lambda)^2 \geq \norm{v}^2 \geq \sum_{i=1}^p \norm{u_i}^2 > p \lambda^2.$$
Thus, we get that $p < {\alpha}^2$.

Note that each $u_i $ is smaller than $v$ and hence, $\norm{u_i} \leq \alpha \lambda$.
From \cref{cor:TU-circuits}, the number of circuits $u$ of $A$ with $\norm{u} \leq \alpha \lambda$ is $m^{O(\alpha^4)}$.
Thus, there can be at most $m^{O(p \alpha^4)}$ vectors of the form $u_1 + u_2 +\cdots + u_p$.
This gives a bound of $m^{O(\alpha^6)}$ on the number of vectors $v$ with $\norm{v} \leq \alpha \lambda$.
\end{proof}

\section*{Acknowledgements}
We thank Rohan Kapadia, Thomas Thierauf  for helpful discussions,
and Ben Lee Volk for pointing out that random binary linear codes 
have exponentially many near-minimum weight codewords. 

We thank Weihao Zhu for identifying an error in
Claim~\ref{cla:no-circuit} of an earlier version
of this paper, and for 
 discussing a repair with us, in particular, for suggesting the use of the strengthened
3-sum decomposition property of Aprile and
Fiorini~\cite{AF22}.
The corrected proof uses this property together with a
$\Delta$--$Y$ transformation.

\bibliographystyle{plain}
\bibliography{tum}

\end{document}